\documentclass[sigconf]{acmart}

\usepackage{booktabs} 

\copyrightyear{2019}
\acmYear{2019}
\setcopyright{iw3c2w3}
\acmConference[WWW '19 Companion]{Companion Proceedings of the 2019 World Wide Web Conference}{May 13--17, 2019}{San Francisco, CA, USA}
\acmBooktitle{Companion Proceedings of the 2019 World Wide Web Conference (WWW '19 Companion), May 13--17, 2019, San Francisco, CA, USA}
\acmPrice{}
\acmDOI{10.1145/3308560.3317587}
\acmISBN{978-1-4503-6675-5/19/05}

\usepackage{balance}

\usepackage[show]{notes-alt}
\usepackage{algorithm}
\usepackage{algorithmic}
\usepackage{amsmath}
\usepackage{booktabs}
\usepackage{mathtools}
\usepackage{soul}
\usepackage{xcolor}
\usepackage{graphicx}
\graphicspath{ {./images/} }
\usepackage{adjustbox}

\usepackage[hang,tight]{subfigure}  
\usepackage{booktabs}               
\usepackage{paralist}               
\usepackage{xspace}

\usepackage[normalem]{ulem} 
\newcommand{\stkout}[1]{\ifmmode\text{\sout{\ensuremath{#1}}}\else\sout{#1}\fi}

\def\pprw{8.5in} \def\pprh{11in}

\newtheorem{problem}{Problem}
\theoremstyle{empty}

\newcommand{\spara}[1]{\smallskip\noindent{\bf #1}}

\newsavebox{\mybox}
\newenvironment{framed}{\noindent\begin{lrbox}{\mybox}\begin{minipage}{.98\columnwidth}}{\end{minipage}\end{lrbox}\fbox{\usebox{\mybox}}}

\newcommand{\Card}[1]{\left\lvert#1\right\rvert}

\newcommand{\St}{;\:}

\newcommand{\skillset}{\ensuremath{{S}}\xspace}
\newcommand{\taskset}{\ensuremath{\mathcal{J}}\xspace}
\newcommand{\workerset}{\ensuremath{\mathcal{W}}\xspace}

\newcommand{\nonhiredworkers}{\ensuremath{\mathcal{F}}\xspace}

\newcommand{\numskills}{\ensuremath{m}\xspace}

\newcommand{\numworkers}{\ensuremath{n}\xspace}

\newcommand{\atask}{\ensuremath{J}\xspace}

\newcommand{\worker}[1]{\ensuremath{W^{#1}}\xspace}
\newcommand{\aworker}{\ensuremath{W}\xspace}

\newcommand{\maxhiringcost}[1]{\ensuremath{C^*}\xspace}

\newcommand{\maxoutsourcecost}[1]{\ensuremath{\lambda^*}\xspace}

\newcommand{\maxhiringcostwithsalary}[1]{\ensuremath{\widehat{C}^*}\xspace}

\newcommand{\solution}{\ensuremath{\mathcal{Q}}\xspace}

\newcommand{\pool}[1]{\ensuremath{P_{#1}}\xspace}

\newenvironment{enumerate-algo}%
{\begin{list}{\arabic{enumi}.}%
      {\setlength{\leftmargin}{2.5em}%
       \setlength{\itemsep}{-\parsep}%
       \usecounter{enumi}}%
}{\end{list}}

\begin{document}

\title{Algorithms for Fair Team Formation in Online Labour Marketplaces}
\titlenote{The research for this work has been partially supported by the 
EU FET project MULTIPLEX 317532 and 
the ERC Advanced Grant 788893 AMDROMA "Algorithmic and Mechanism Design Research in Online Markets"
}

\author{Giorgio Barnab\`o}
\affiliation{
  \institution{Sapienza University of Rome, Italy}
}
\email{barnabo@diag.uniroma1.it}

\author{Adriano Fazzone}
\affiliation{
\institution{Sapienza University of Rome, Italy}
}
\email{fazzone@diag.uniroma1.it}

\author{Stefano Leonardi}
\affiliation{
\institution{Sapienza University of Rome, Italy}
}
\email{leonardi@diag.uniroma1.it}

\author{Chris Schwiegelshohn}
\affiliation{
\institution{Sapienza University of Rome, Italy}
}
\email{schwiegelshohn@diag.uniroma1.it}

\renewcommand{\shortauthors}{Barnab\`o et al.}

\begin{abstract}
As freelancing work keeps on growing almost everywhere due to a sharp decrease in communication costs and to the widespread of Internet-based labour marketplaces (e.g., \textit{guru.com}, \textit{feelancer.com}, \textit{mturk.com}, \textit{upwork.com}), many researchers and practitioners have started exploring the benefits of outsourcing and crowdsourcing ~\cite{howe06, jeppesen10, kittur11, malone10, retelny14, surowiecki04}. Since employers often use these platforms to find a group of workers to complete a specific task, researchers have focused their efforts on the study of team formation and matching algorithms and on the design of effective incentive schemes ~\cite{lappas2009finding, anagnostopoulos2010power, anagnostopoulos12online, Anagnostopoulos2018AlgorithmsFH}. Nevertheless, just recently, several concerns have been raised on possibly unfair biases introduced through the algorithms used to carry out these selection and matching procedures. For this reason, researchers have started studying the fairness of algorithms related to these online marketplaces ~\cite{Feldman2015CertifyingAR, levy2017designing}, looking for intelligent ways to overcome the algorithmic bias that frequently arises. Broadly speaking, the aim is to guarantee that, for example, the process of hiring workers through the use of machine learning and algorithmic data analysis tools does not discriminate, even unintentionally, on grounds of nationality or gender. \\
In this short paper, we define the \textit{Fair Team Formation problem} in the following way: given an online labour marketplace where each worker possesses one or more skills, and where all workers are divided into two or more not overlapping classes (for examples, men and women), we want to design an algorithm that is able to find a team with all the skills needed to complete a given task, and that has the same number of people from all classes. \\
We provide inapproximability results for the \textit{Fair Team Formation problem} together with four algorithms for the problem itself. 
We also tested the effectiveness of our algorithmic solutions by performing experiments using real data from an online labor marketplace.
\end{abstract}

%
%




\maketitle

\section{Introduction}\label{sec:intro}
An online labour marketplace is defined as a web application where workers can sell their services and skills in a fluid and delocalised fashion. Usually, employers pay workers hourly to complete a specific task without offering them any long-term employment arrangement. The OECD data on self-employment estimates that between 10\% and 20\% of workers in developed countries are self-employed, while it is estimated that in 2020, a full 40\% of the US workforce will be freelancers ~\cite{NextBigThin}. \\
While crowdsourcing adoption was driven, at least in part, by the assumption that problems can be decomposed into parts that can be addressed separately by independent workers, recent work suggests that crowdsourcing results can be improved by allowing some degree of collaboration among them~\cite{majchrzak2013towards, Riedl2016}. The idea of combining collaboration with crowdsourcing has led to research on \emph{Team Formation}~\cite{anagnostopoulos2010power, anagnostopoulos12online, an13finding, dorn10composing, gajewar12multiskill, golshan14profit, kargar11teamexp, lappas2009finding, li10team, sozio10community, majumder12capacitated}, in which a common thread is the need for complementary skills, and definitions differ in aspects such as objectives (e.g., load balancing and/or compatibility), constraints (e.g., worker capacity), and algorithmic set-up (online or offline). \\
As previously mentioned, these online marketplaces are largely managed through automatic algorithms designed to match supply and demand. Nevertheless, the objective of optimising a given task, which these algorithms are usually based on, goes openly against the need to ensure fairness and diversity, for example, in the composition of groups. We define \textit{unfair discrimination} as treating someone differently on the base of his group membership, and not his merit. Since algorithms are "\textit{black boxes}" usually protected by industrial secrecy, legal protections and even intentional obfuscation, most of the times discrimination becomes invisible, and mitigation impossible~\cite{Hajian:2016:ABD:2939672.2945386}. For this reason, data scientists and researchers have developed the \textit{disparate impact} theory~\cite{Feldman2015CertifyingAR} whose aim is to spot unintended discrimination in algorithms outcomes. Among the many different sources of the bias on the Web, the one that directly concerns us in the research of a solution for the \textit{Fair Team Formation problem} is the \textit{algorithmic bias}, that occurs when the bias is added by the algorithm itself or by the way this algorithm manages the bias present in the data it crunches.

\spara{Overview of problem setting and assumptions.} In our framework, both workers and tasks are represented by sets of skills. Each skill of the task is possessed by at least one worker, while each worker has a defined cost and belongs alternately to one of two classes. In this setting, we consider the problem of finding the cheapest team of workers that together have all the necessary skills to complete the task, and that is made up of the same number of workers from both classes (fairness constrains). We call this general problem \textit{Fair Team Formation}, which we formally define in Section 2 and solve in Sections 3-4.

\spara{Algorithmic techniques.} To the best of our knowledge, we are the first to consider the \textit{Fair Team Formation} problem, namely the weighed Set Cover problem with some fairness constraints imposed. As shown in section 3, the \textit{Fair Team Formation} is NP-hard and inapproximable, for this reason the only thing we could do was to look for some algorithms that would function well in practical situations. Now, considering that our problem is closely related to the Set Cover problem ~\cite[Chapter~35]{cormen2009introduction}, it seemed natural to start from a reasoning similar to the one behind the Greedy Set Cover algorithm ~\cite{slavik1997tight}. In the next sections, then, we will present four algorithms we developed to solve the \textit{Fair Team Formation} problem: the first three are partially based on the Greedy Set Cover algorithm, while the fourth is a rounding algorithm based on the linear programming formulation of the Fair Team Formation problem. \\
Furthermore, since we are not able to calculate the value of the optimal solution in reasonable time, we have built a lower-bound for the cost of the optimal solution of the \textit{Fair Team Formation} problem by solving the relaxed Linear Programming formulation of our problem. This lower-bound came in handy when we had to evaluate our algorithms performance.  

\spara{Contributions.}
The key contributions of our work are:
\begin{compactitem}

\item We formalise the \textit{Fair Team Formation} problem, which is the problem  
of finding the cheapest team that can complete the task and, at the same time, that counts the same number of people from two not overlapping classes.

\item We design four algorithms for solving our problem.

\item We experiment on real data based on actual task requirements and worker skills from one of the largest online labor marketplaces, testing algorithms under a broad range of conditions.

\end{compactitem}

\section{Preliminaries}\label{sec:preliminaries}
In this section, we formally describe our setting and problem, and provide some necessary background.

\subsection{Notation and Setting}

\spara{Skills.} We consider a set $\skillset$ of skills with $\Card{\skillset}=\numskills$. Skills can be any kind of qualification a worker can have or a task may require, such as \emph{video editing}, \emph{technical writing}, or \emph{project management}.

\spara{Tasks.} We consider a set of $\taskset$ tasks (or jobs). Each task $\atask\in\taskset$ is independent, and requires a set of skills from $\skillset$, therefore, $\atask\subseteq \skillset$. In our setting we do not consider a streaming of tasks, but rather we take each task as a single instance of the problem. 

\spara{Workers.} Throughout we assume that we have a set $\workerset$ of $\numworkers$ workers: $\workerset = \{\worker{r}\St r = 1,\dots,\numworkers \}$. Every worker $r$ possesses a set of skills ($\worker{r}\subseteq\skillset$). Similarly to the tasks, we use $\worker{r}$ to denote both the worker and his/her skills. Moreover, each worker has a hiring cost, and belongs alternately to one of two not overlapping classes. 

\spara{Classes.} The workforce is split in two not overlapping classes $C = \{class1, class2\}$, for example women and men.

\spara{Coverage of tasks.} Whenever task $J \subseteq\skillset$ arrives, an algorithm has to assign one or more workers to it, i.e., a \emph{team}. We say that $J$ can be \emph{completed} or \emph{covered} by a team $\solution\subseteq\workerset$ if for every skill required by $J$, there exists at least one worker in $\solution$ who possesses this skill: $J \subseteq \cup_{\aworker\in\solution}\aworker$. We assume that for every skill in the incoming task there is at least one worker possessing that skill, so all tasks can be covered.

\begin{table}[t]
\caption{Notation}
\label{tbl:notation}
\centering\small\begin{tabular}{cl}
\toprule
$\skillset$	& Set of skills, size $\numskills$ \\
$\taskset$	& Set of tasks \\
$\workerset$	& Set of workers, size $\numworkers$. \\
		& $\worker{r}_\ell=1$ if worker $r$ possess skill $\ell$, 0 otherwise\\
$\pool{\ell}$	& Subset of workers possessing skill $\ell$ \\
$C$ & the two classes $\{class1, class2\}$ \\
\bottomrule
\end{tabular}
\end{table}

\subsection{Problem Definition}

We now define the problem that we study:

\begin{problem}[The Fair Team Formation problem]\label{problem:theproblem} There exists a set of skills \skillset. We have a pool of workers \workerset, where each worker $\worker{r} \in\workerset$ is characterised by \begin{inparaenum}[(1)] a subset of skills $\worker{r}\subseteq\skillset$, a hiring cost $c_r\in\mathbb{R}_{\ge0}$, and belongs alternately to one of two not overlapping classes, $C = \{class1, class2\}$.
\end{inparaenum}
Given a task $\atask\in\taskset$, the goal is to design an algorithm that, when task $\atask$
arrives, decides which workers to hire such that all the tasks are covered by the workers who are hired, the total cost paid over all the tasks is minimised, and the team formed is made up of the same number of workers from both classes, $C = \{class1, class2\}$. 
\end{problem}

One special case of the \textit{Fair Team Formation problem}, where no fairness constraints are imposed, is the \textit{Weighed Set Cover problem}. This problem can be effectively addessed through a greedy approach (see~\cite[Chapter 2]{vazirani2013approximation}). As shown by Slavik ~\cite{slavik1997tight}, this greedy algorithm has an approximation ratio of $\log n -\log\log n +\Theta(1)$ ~\cite{slavik1997tight}. Unfortunately, this result does not hold true for the Fair Set Cover that is the algorithmic core behind the Fair Team Formation problem.


\section{Inapproximability \& Lower-Bound}\label{sec:useforfree}
\newcommand{\poolUnhired}[1]{\ensuremath{P^{\nonhiredworkers}_{#1}}\xspace}
\newcommand{\poolUnhiredUseful}[1]{\ensuremath{P^{\nonhiredworkers}_{#1}}\xspace}

First, we will show that the Fair Team Formation problem (i.e. the Fair Weighted Set Cover problem) is inapproximable, then we will present two different lower-bounds that we can easily calculate, and use later on to evaluate the quality of the solutions found by our four algorithms. 

\subsection{Inapproximability of the Fair Set Cover Problem}
\label{subsec:simple-naive}

Cover problems on hyper-graphs $H(V, E, w)$ aim to find a subset $S \subset E$ such that $v \in \cup_{S_i \in S} S_i$ for every $v \in V$ and \textit{w}(S) is minimised. The vertex cover problem is a special case where we are given a graph $G(V', E')$ and aim to find a subset $S' \subset E'$ such that every edge $e \in E$ is incident to at least one node of $S'$. In terms of hyper-graphs, $V$ corresponds to $E'$ and each hyper-edge $h_v \in H$ corresponds to the set of edges incident to $v$. \\
Given a coloring $c:V\rightarrow \{red,blue\}$ of $G$, we consider a set of vertexes $S\subset V$ to be fair, if $|S\cap \text{RED}| = |S\cap \text{BLUE}|$.
The fair vertex cover problem consists of finding a  minimum vertex cover under the constraint that it is fair. Note that unlike the unconstrained fair vertex cover, such a set may not exist in general.
Similarly, given a coloring of the sets $c:E\rightarrow \{red,blue\}$ the fair set cover problem consists of finding a minimum set cover $S\subset E$ such that $|S\cap \text{RED}| = |S\cap \text{BLUE}|$.
We note that generally fair covers need not exist.
This feature will allow us to show the following impossibility result.

\begin{theorem}
Computing any finite approximation of the fair vertex cover problem is NP-hard.
\end{theorem}
 
 \begin{proof}
Let $G(V, E)$ be a graph, where we consider $V$ to be red. Given an integer $k$, it is NP-hard to determine whether there exists a vertex cover of size at most $k$~\cite{GareyJ79}. We add $k$ blue vertexes $V'$.
If there exists a fair vertex cover in $G'(V\cup V',E)$, then it can consist of at most $k$ blue vertexes. Since any finite approximation of the fair vertex cover algorithm in particular determines the existence of a fair vertex cover, it also solves the decision problem of vertex cover.
Hence, computing any finite approximation of the vertex cover is NP-hard.
\end{proof}
 
\begin{corollary}
Computing any finite approximation of the fair set cover problem or the fair group Steiner tree problem is NP-hard. 
\end{corollary}
\begin{proof}
Both problems contain the vertex cover problem as a special case~\cite{GareyJ79,ReichW89}.
\end{proof}

Finally, it is worth noting that for the unweighted version of the Fair Set Cover problem (i.e. all workers have the same cost), and under the \textit{enough workers} assumption, we can build a simple algorithm whose approximation factor is equal to $|C|H(|T|)$. 

\subsection{Lower-bound}

When trying to solve an instance of the Fair Set Cover problem, we are often unable to calculate the value of the optimal solution in reasonable time;  therefore, we are forced to use algorithms that find only a suboptimal solution to the problem. For this reason, it is important to have a lower-bound which we are sure that the value of the optimum would never go below. Obviously, a first really trivial lower-bound (TLB) is represented by the cost of the solution we obtain when the Greedy Set Cover is applied to the Fair Set Cover instance (after eliminating the fairness constraints), divided by its approximation factor; namely: 

\begin{equation}
TLB = \frac{Cost(Greedy Set Cover Solution)}{\log(n) - \log(\log(n)) + 3. + \log(\log(32)) - \log(32)}
\end{equation}  

\spara{A Lower-Bound from the Relaxed LP formulation of the Fair Set Cover problem.} A computationally feasible and mathematically elegant way to calculate a better lower-bound for the Fair Set Cover problem is to solve its relaxed Linear Programming formulation. In a nutshell, we formulate the Fair Team Formation problem as an Integer Linear Programming problem, and then we relax its constraints. In this way, we obtain a Linear Programming problem that is solvable in polynomial time and whose solution always costs no more than the optimal solution that we would get if we were able to solve the integer linear programming. 
The Relaxed Linear Programming formulation of the Fair Set Cover problem is the following:


\begin{framed}
\label{box:LPF}
Relaxed Linear program for  the Fair Set Cover problem:
$$
\begin{cases} 
min \quad \displaystyle \sum_{i = 1}^{|W|} c_ix_i &\\ 
s.t. \quad \displaystyle \sum_{i: s \in W_i} x_i \geq 1  \quad & \forall s \in T \\ 
and \quad x_i \in \left[0, 1 \right] &  \forall i \in \{1, ..., |W|\} \\
and \quad \displaystyle \sum_{i = 1}^{|W|} k_ix_i = 0
\end{cases} 
$$
Where $x_i$ assumes either value 0 or 1, depending on whether the $i_{th}$ worker is hired or not; $c_i$ is the worker $i$ hiring cost, and $k_i$ is equal to -1 if worker belongs to $class1$, or to 1 if he belongs to $class2$. 
\end{framed}

\section{The Fair Team Formation problem}\label{sec:salarynew}
Given the previous restrictive result, in this section, we provide four algorithms to solve the Fair Team Formation problem. Considering that the Fair Team Formation problem has a lot in common with the Set Cover problem, it seemed natural to start from a reasoning similar to the one behind the Greedy Set Cover algorithm. Therefore, algorithms~\ref{alg:padding},~\ref{alg:alternating},~\ref{alg:pairs} are partially based on the Greedy Set Cover algorithm, while algorithm~\ref{alg:rounding} is a rounding algorithm based on the linear programming formulation of the Fair Team Formation problem. The only assumption we made is that there is always a team of workers that together have all the necessary skills to complete the task we are handling. In other words, the task is always coverable.

\spara{Fair Padding Greedy Set Cover algorithm.} The first algorithm we came up with is a simple extension of the Greedy Set Cover algorithm where the cheapest workers of the class whose cardinality is lower are added to make the team fair. Algorithm~\ref{alg:padding} shows its pseudocode.

\begin{algorithm}[H]
    \caption{FairPaddingGreedySetCoverAlgorithm}
    \label{alg:padding}
    \textbf{Input}: $(\textbf{W}, J)$. \\
    \textbf{Ouput}: FairTeam $W \subseteq \textbf{W}$.
 \begin{algorithmic}[1]
    \STATE $W_0 \leftarrow$ GreedySetCoverAlgorithm(W, J)
    \IF{$W_0$ is not balanced (i.e. different number of workers from the two classes)}
        \STATE $W_1 \leftarrow W_0\cup GetCheapestWorkers(W, W_0, MinorityClassCard.)$
         \STATE Return $W_1$
    \ELSE{}
        \STATE Return $W_0$
    \ENDIF 
\end{algorithmic}
\end{algorithm}

The time complexity of this algorithm is equal to the time complexity of the Greedy Set Cover algorithm, namely: $O(|\textbf{W}||J|^2)$.

\spara{Fair Alternating Greedy Set Cover algorithm.} Let's start by defining the marginal utility of each worker (WMU) as:

\begin{equation}
WMU = \frac{WorkerCost}{|SetOfWorkerSkills \cap SetOfTaskSkillsNotCoveredYet|}
\end{equation}

\medskip

Heuristically, at each stage, the AlternatingGreedySetCoverAlgorithm chooses the worker with the lower marginal utility alternating the class of workers within which it picks. Algorithm~\ref{alg:alternating} shows its pseudocode.

\begin{algorithm}[H]
    \caption{FairAlternatingGreedySetCoverAlgorithm}
    \label{alg:alternating}
     \textbf{Input}: $(W, J)$. \\
     \textbf{Ouput}: FairTeam $W \subseteq \textbf{W}$.
 \begin{algorithmic}[1]
    \STATE $W_1\leftarrow$ AlternatingSetCover(W, J, StartingClass = 1)
    \IF{$W_1$ is not balanced (i.e. different number of workers from the two classes)}
        \STATE $W_1 \leftarrow W_1 \cup GetCheapestWorkers(W, W_1, MinorityClassCard.)$
    \ENDIF
    \STATE $W_2$ $\leftarrow$ AlternatingSetCover(W, T, StartingClass = 2)
    \IF{$W_2$ is not balanced (i.e. different number of workers from the two classes)}
        \STATE $W_2 \leftarrow W_2 \cup GetCheapestWorkers(W, W_2, MinorityClassCard.)$
    \ENDIF
    \STATE Return the cheapest team between $W_1$ and $W_2$
\end{algorithmic}
\end{algorithm}

Also in this case, the time complexity is: $O(|\textbf{W}||J|^2)$.

\spara{Fair Pairs Greedy Set Cover algorithm.} Algorithm~\ref{alg:pairs} is particularly simple and intuitive. Essentially, it is the application of the Greedy Set Cover algorithm to all possible pairs of workers. This idea has been suggested by~\cite{Fairlets}. Algorithm~\ref{alg:pairs} shows its pseudocode.

\begin{algorithm}[H]
    \caption{FairPairsGreedySetCoverAlgorithm}
    \label{alg:pairs}
     \textbf{Input}: $(W, J)$. \\
     \textbf{Ouput}: FairTeam $W \subseteq \textbf{W}$.
 \begin{algorithmic}[1]
    \STATE $W_{Pairs} \leftarrow$ PairsGenerator(W)
    \STATE $W_0 \leftarrow$ CoupleGreedySetCover($W_{Pairs}$, J)
    \STATE Return $W_0$
\end{algorithmic}
\end{algorithm}

Unlike the previous three algorithms, in this case the time complexity is: $O(|\textbf{W}|^2|J|^2)$. 
The $|\textbf{W}|^2$ factor is due to the fact that the greedy algorithm for the set cover problem has as input the set of all unordered couples of workers.

\spara{Relaxed Fair Set Cover Rounding algorithm.} Algorithm~\ref{alg:rounding} solves the relaxed linear programming formulation of the Fair Team Formation problem assigning to each worker a real number between 0 and 1: this number could be interpreted as the worker's probability to be hired. Then, it continues by creating random teams of workers using these probabilities until it finds a team that is both fair and able to complete the task. 

\begin{algorithm}[H]
    \caption{RelaxedFairSetCoverRoundingAlgorithm}
    \label{alg:rounding}
     \textbf{Input}: $(\textbf{W}, J)$. \\
     \textbf{Ouput}: FairTeam $W \subseteq \textbf{W}$.
 \begin{algorithmic}[1]
    \STATE $Hiring_{ProbabilityVector}\leftarrow$FairTeamFormationRelaxedLP(\textbf{W}, J)
    \STATE $\textbf{W}_{Sorted} \leftarrow SortAccordingToProbabilityVector(\textbf{W})$
    \WHILE{$\neg$(W balanced $\wedge$ task skills are all covered)}
        \STATE $W \leftarrow EmptyTeam$
        \FOR{$w \in \textbf{W}_{Sorted}$}
	\STATE add w to W with probability equal to $Hiring_{ProbabilityVector}(w)$
	\IF{(W balanced $\wedge$ task skills are all covered)}
	\STATE Return W
	\ENDIF
    \ENDFOR
    \ENDWHILE
    \STATE Return W
\end{algorithmic}
\end{algorithm}

\section{Experiments}\label{sec:experiments}
In this section, we will present some experiments that we ran on a real dataset to evaluate the algorithms' performance by comparing their cost.

\subsection{The Freelancer dataset}\label{subsec:datasets}

To create a large pool of tasks and workers needed to test the algorithms, we decided to use a dataset obtained from \textit{Freelancer.com}: the largest online marketplace for outsourcing in its category according to data from Alexa (Feb. 2018). The input data that we obtained contain anonymised profiles from people registered as freelancer in this marketplace. This includes their self-declared sets of skills, as well as the average rate that they charge for their services. Data have been cleaned to remove skills that were not possessed by any worker, and skills that were never required by any task. Concerning tasks, we had access to a large sample of tasks commissioned by buyers in the marketplace. Some relevant characteristics of our data are summarised in table~\ref{tab:title}.

\begin {table}[H]
\caption {Freelancer Dataset Characteristics} \label{tab:title} 
\begin{center}
\begin{tabular}{llr}
\hline
\textbf{Dataset} & \textbf{Freelancer} \\
\hline
Skills (\textit{m})    &  175     \\
Workers (\textit{n})  &  1,211       \\
Tasks (\textit{T})  &  992       \\
...distinct  &  600       \\
\hline
Average Skills/Worker & 1.45 \\
Average Skills/Task & 2.86 \\
\end{tabular}
\end{center}
\end {table}

As shown above, our dataset contains 992 tasks, but since many of them require exactly the same set of skills we decided to take into account only the 600 distinct tasks. The average number of skills per worker is 2.86 and the maximum is 6 skills.

\spara{Experiments Design.} In the first place, we split the 1211 Freelancer workers into two different classes, and we considered six different compositions of the two groups. In brief, we used a random procedure to select respectively 10\%, 30\%, and 50\% of all workers, and we assigned these workers to one of the two classes, while the remaining to the other. After that, for each of these configurations, we ran the four algorithms we designed to solve the Fair Team Formation Problem, 
obtaining fair teams to complete each of the 600 tasks.

\subsection{Experiments}

As shown in figure \ref{fig:cost_random_weighted}, we observe a shift to the left in the distribution, as the workforce becomes more balanced. In most cases the price of the fair team is no more than four times the value of the best lower bound (LB), although for a few tasks the FairAlternatingGreedySetCoverAlgorithm finds solutions that are even eight times the value of the lower bound. It is also worth noting that the progressive balancing of workers' colours has a significant effect on all algorithms, except for the RelaxedFairSetCoverRoundingAlgorithm whose cost (cost\_RLP) distribution remains more or less consistent as the workforce changes. Moreover, from figure \ref{fig:cost_random_weighted} we can see that all distributions are concentrated around a value of 2, indicating that our algorithms have an heuristic approximation ratio of 2, at least on this specific dataset. \\
In summary, histograms in figure \ref{fig:best_random} give us some important information about the overall algorithms performance, obtained by choosing the less expensive fair team among the four on a case-by-case basis. The balance between the two classes of workers does not influence the cost distributions suggesting that some algorithms are able to efficiently address the problem of strong unbalances between the two groups of workers; second, we can observe that the best solution cost is never more than four times the value of its best lower bound, and it rarely exceeds a factor of two. \\
To conclude, the RelaxedFairSetCoverRoundingAlgorithm beats them all: it was able to find a team whose cost is equal to the best solution cost in no less than 66\% of cases, and with an average success rate of 85\% (all configurations of colours considered). On the contrary, the FairPaddingGreedySetCoverAlgorithm always had the worst overall performance, never reaching a success rate higher than 70\%. 


\begin{figure*}[t] 
\centering
\caption{Distribution of solutions cost over best lower bound for three different class balances.}
\label{fig:cost_random_weighted}
\subfigure[Fair Padding Algorithm]{\includegraphics[width=.23\textwidth]{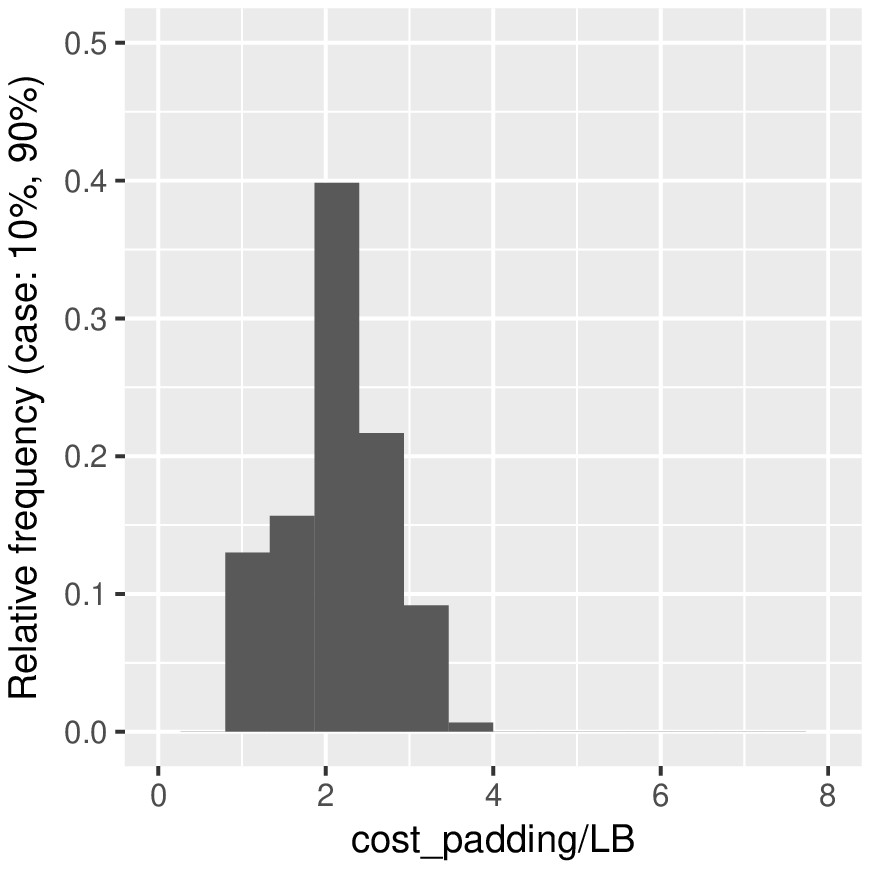}}
\subfigure[Fair Alternating Algorithm]{\includegraphics[width=.23\textwidth]{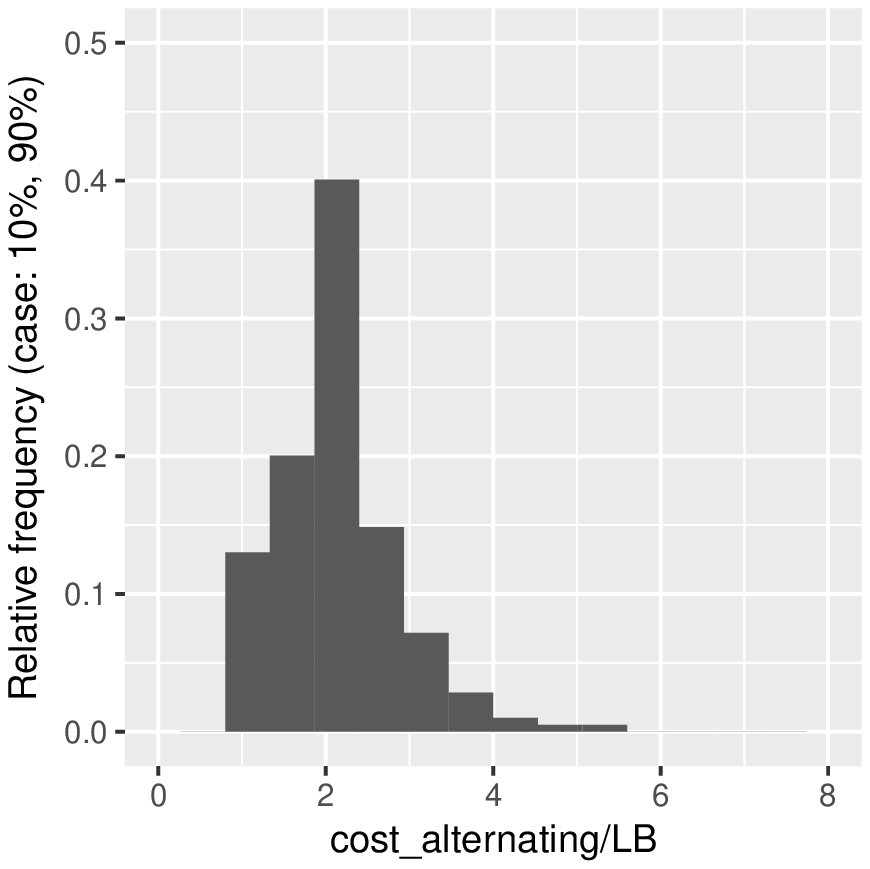}} 
\subfigure[Fair Pairs Algorithm]{\includegraphics[width=.23\textwidth]{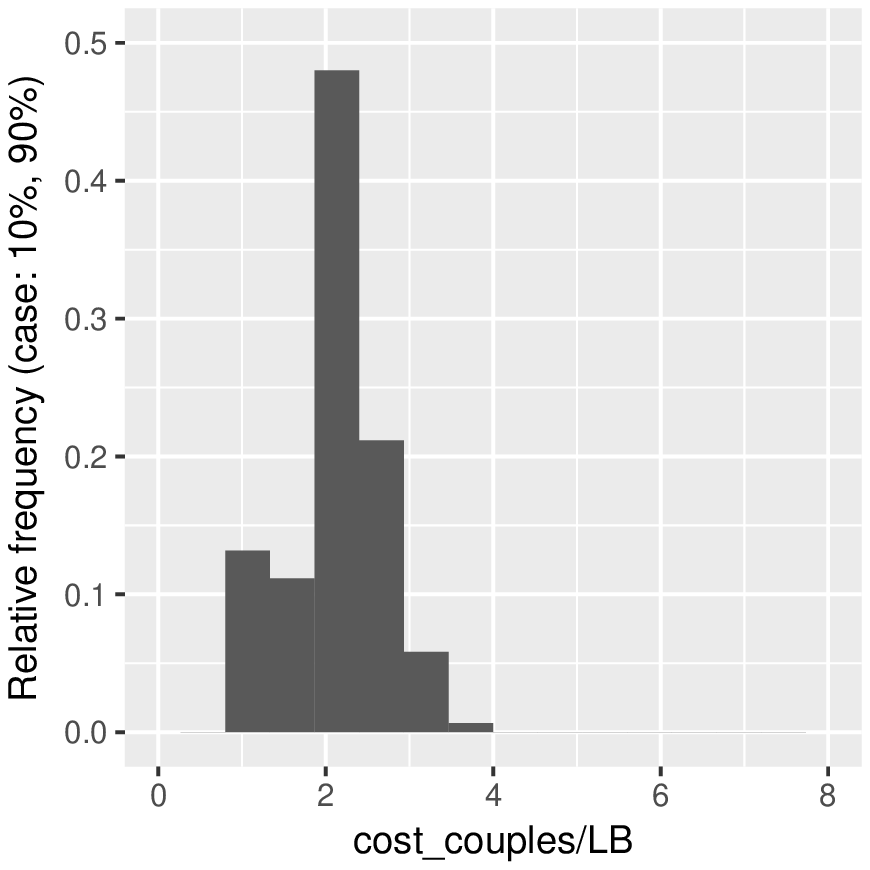}}
\subfigure[Relaxed Fair Rounding Algorithm]{\includegraphics[width=.23\textwidth]{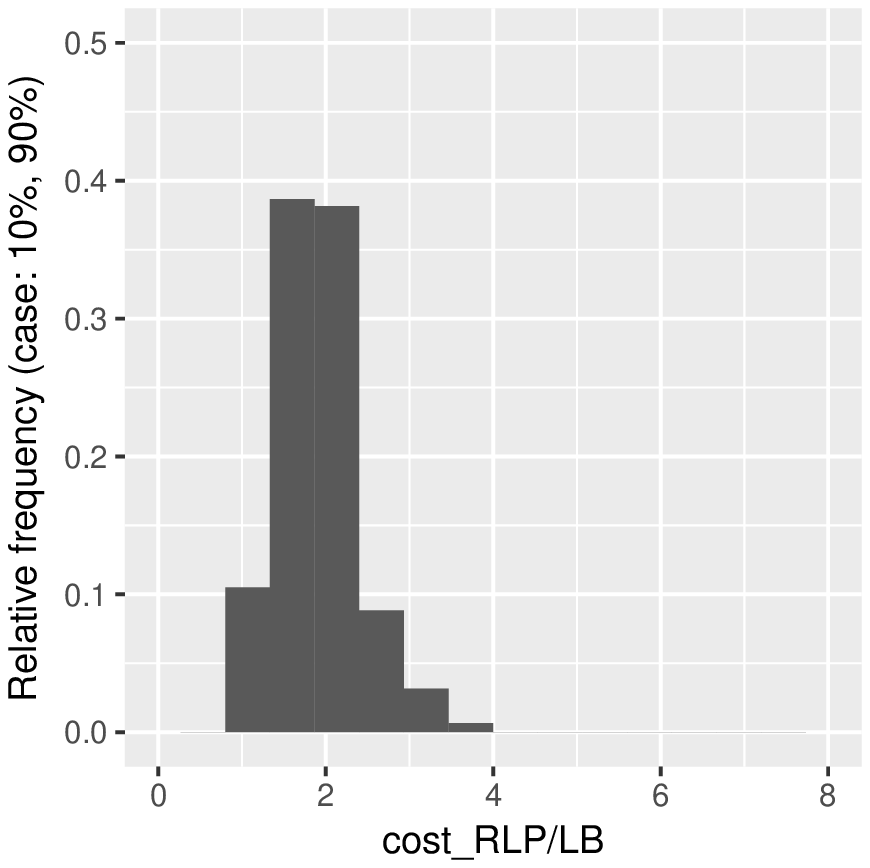}} \\
\subfigure[Fair Padding Algorithm]{\includegraphics[width=.23\textwidth]{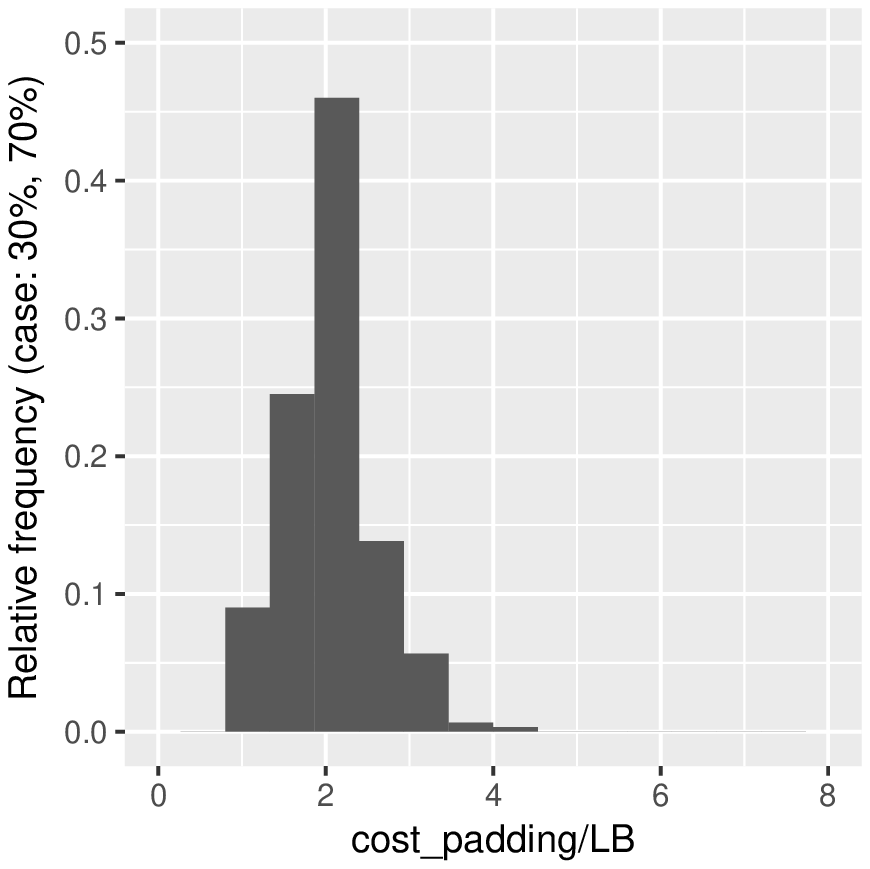}}
\subfigure[Fair Alternating Algorithm]{\includegraphics[width=.23\textwidth]{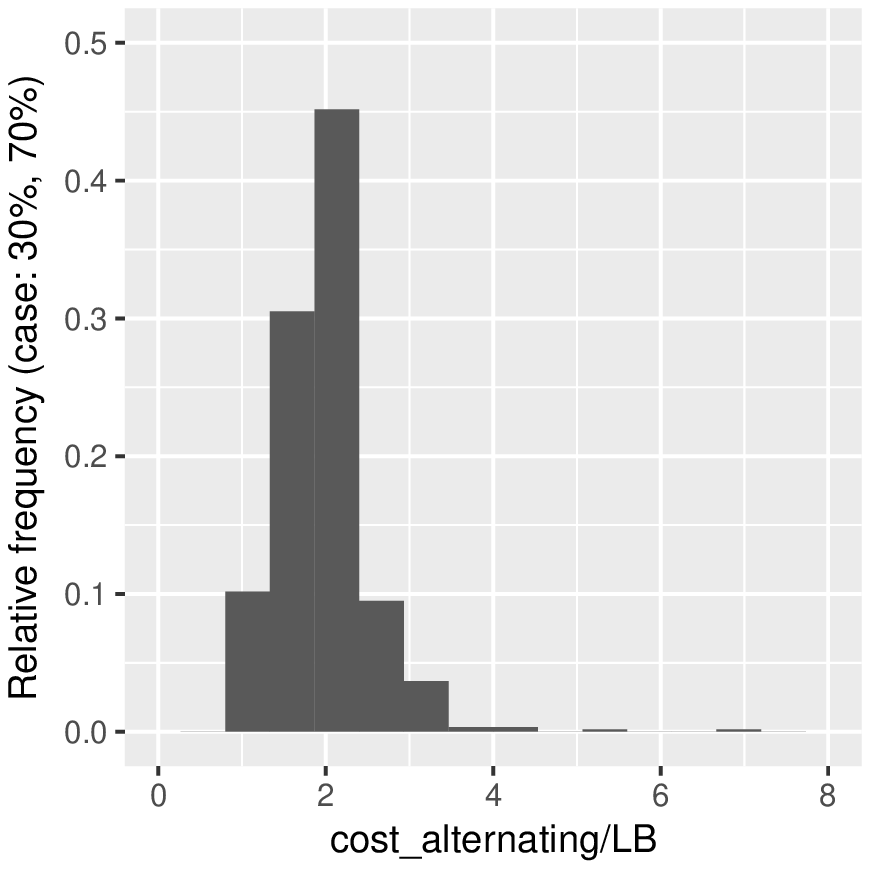}} 
\subfigure[Fair Pairs Algorithm]{\includegraphics[width=.23\textwidth]{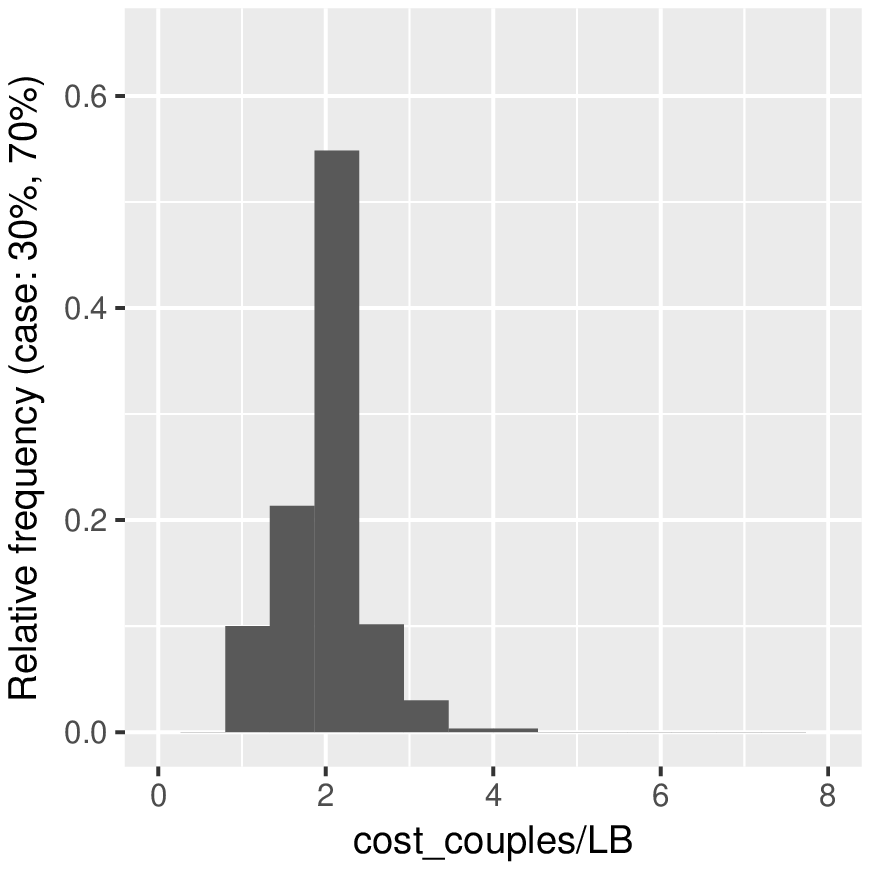}}
\subfigure[Relaxed Fair Rounding Algorithm]{\includegraphics[width=.23\textwidth]{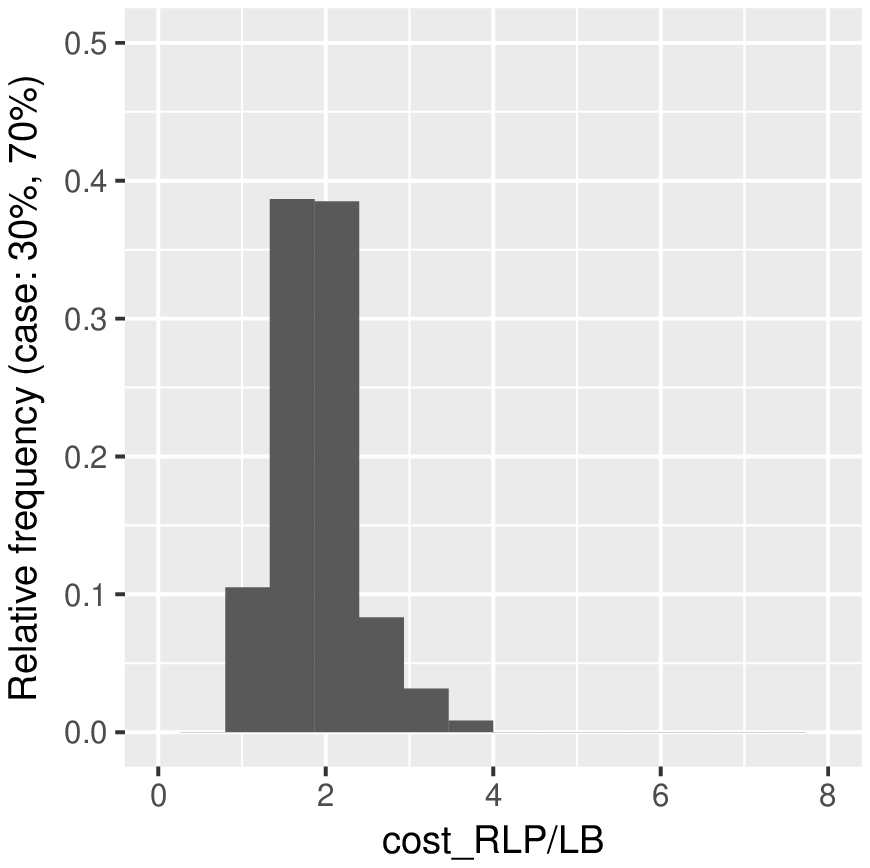}} \\
\subfigure[Fair Padding Algorithm]{\includegraphics[width=.23\textwidth]{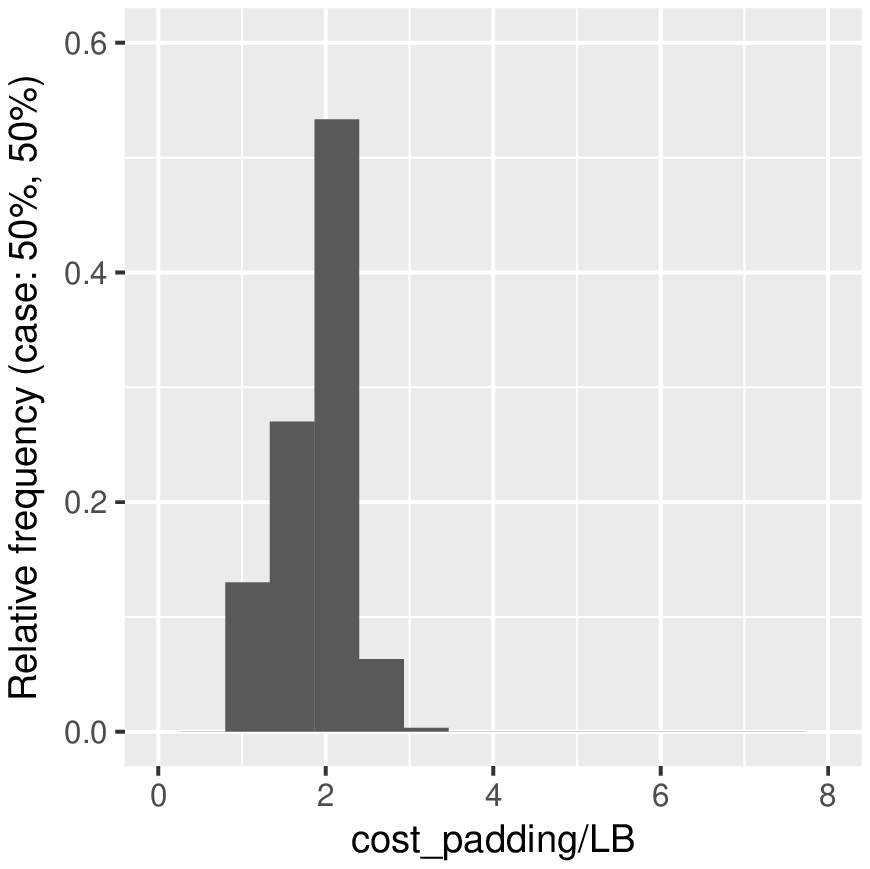}}
\subfigure[Fair Alternating Algorithm]{\includegraphics[width=.23\textwidth]{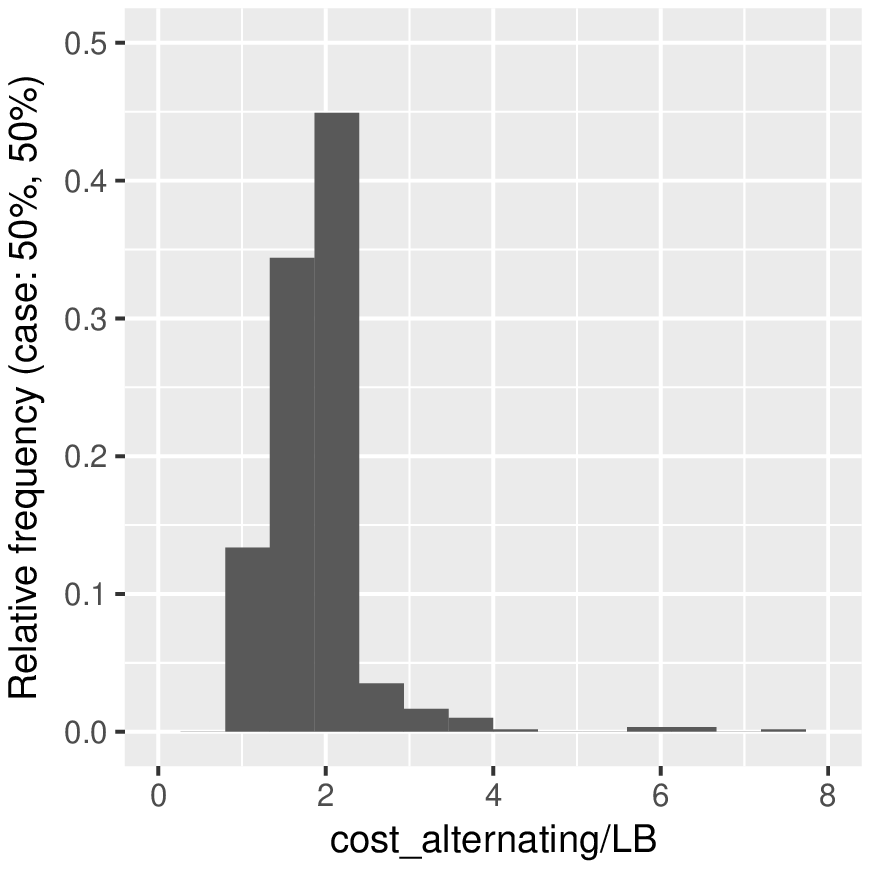}} 
\subfigure[Fair Pairs Algorithm]{\includegraphics[width=.23\textwidth]{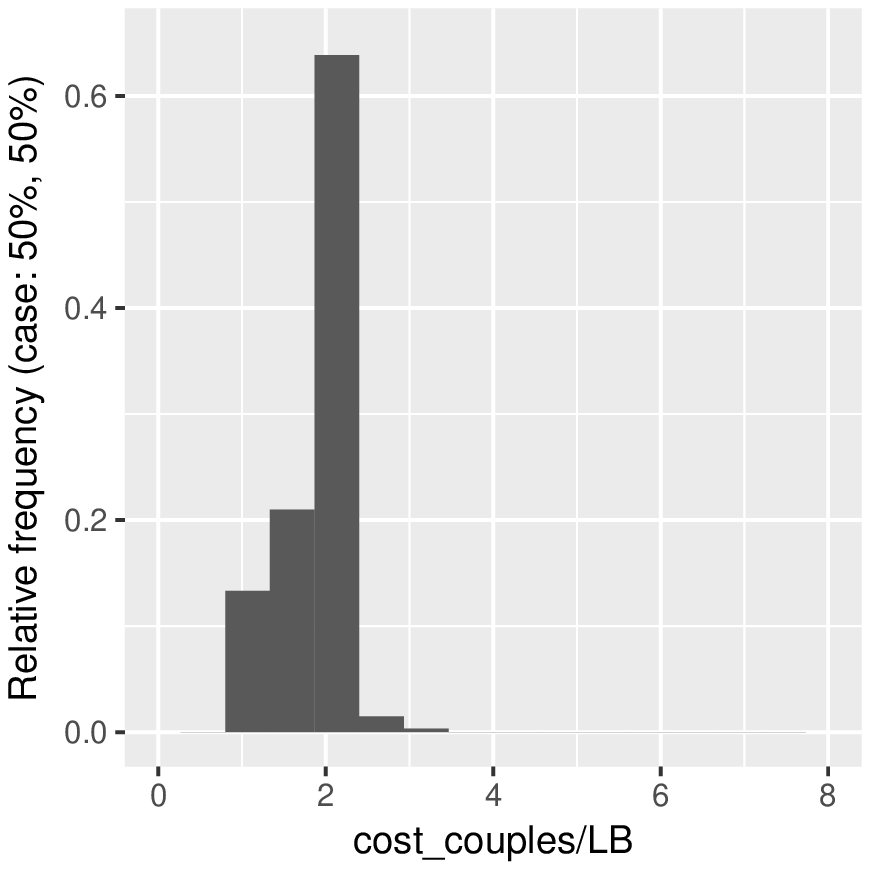}}
\subfigure[Relaxed Fair Rounding Algorithm]{\includegraphics[width=.23\textwidth]{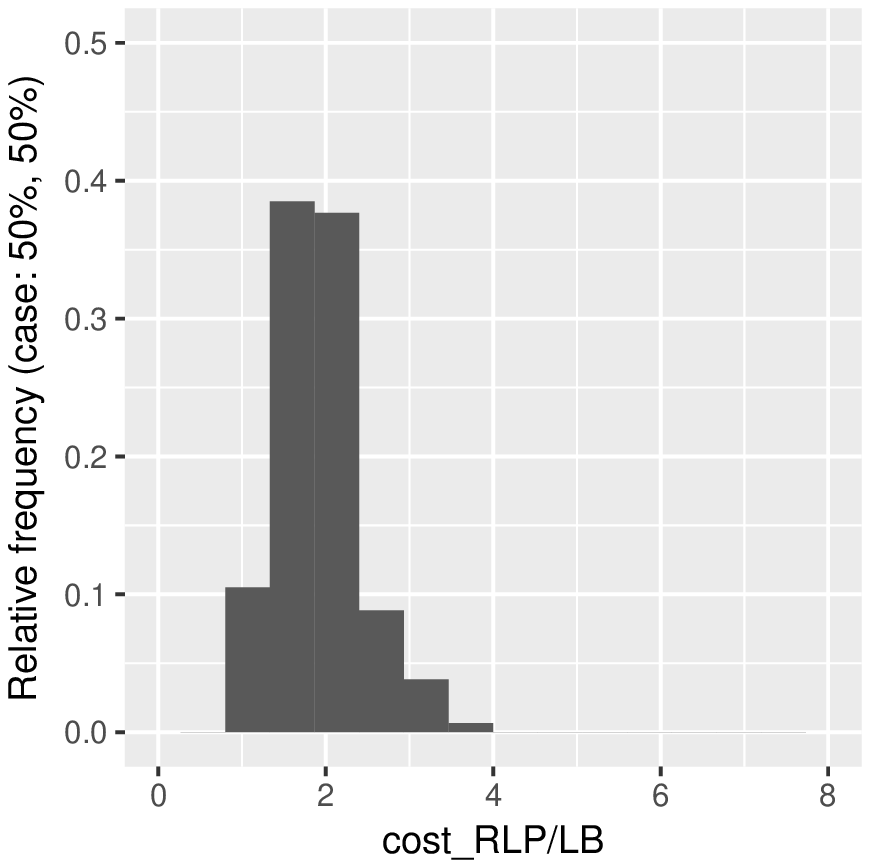}} \\  
\caption{Distribution of best solution cost over best lower bound for three different class balances.}
\label{fig:best_random}
\subfigure{\includegraphics[width=.23\textwidth]{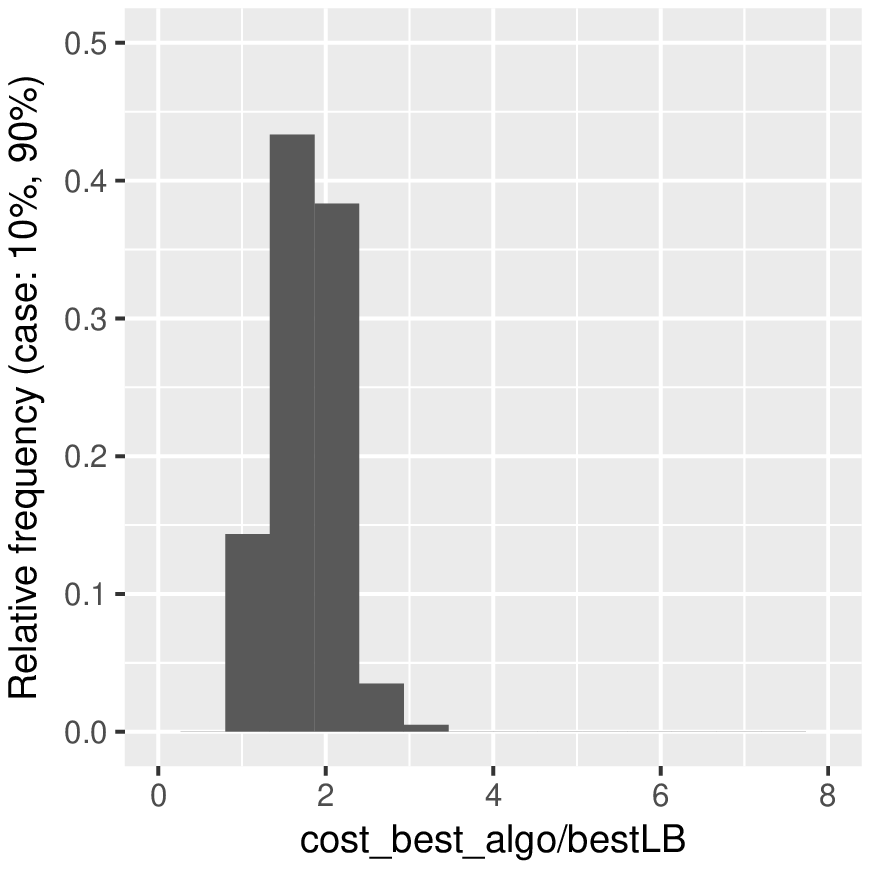}} 
\subfigure{\includegraphics[width=.23\textwidth]{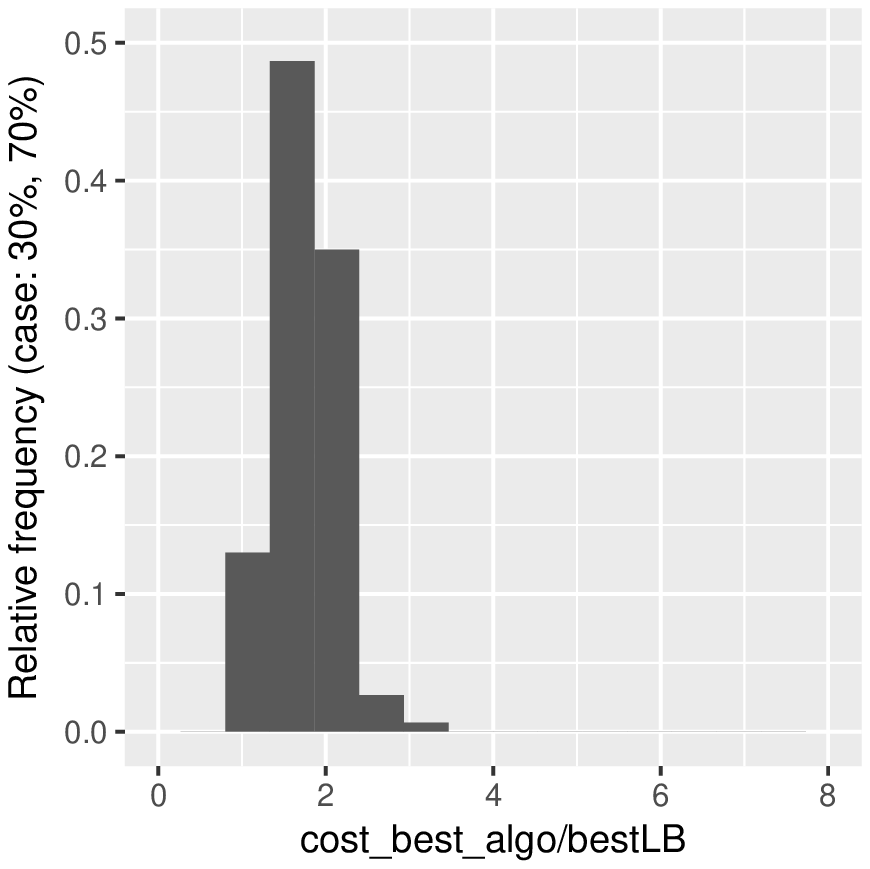}} 
\subfigure{\includegraphics[width=.23\textwidth]{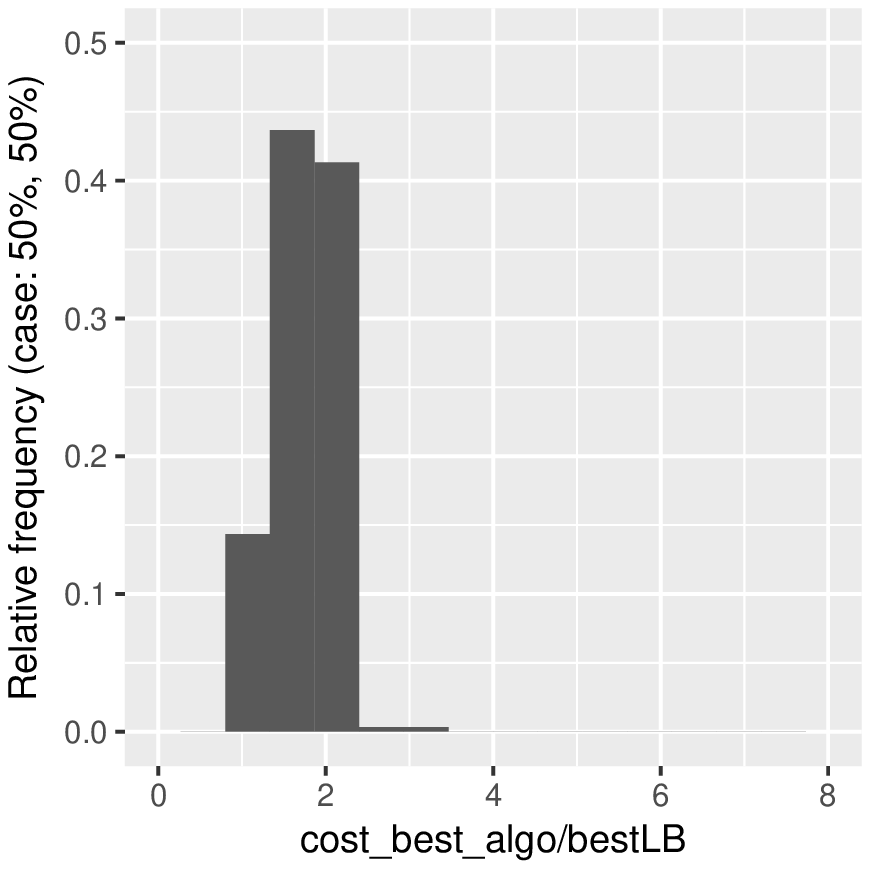}}                                 
\end{figure*}



\section{Conclusions \& Future Work}\label{sec:conclusions}
In this work, we have defined the Fair Team Formation problem, that is a variation of the Set Cover problem where each subset is assigned a colour, and whose goal is to find the cheapest collection of subsets that both covers the input set, and that is made up of the same number of subsets of each colour. \\ 
Despite the discovered inapproximability results, in particular for the Fair Team Formation problem, we have focused our research on the design and implementation of four algorithms for that problem, and we have also tested them on a real dataset. 
From the experiments we conducted on the Freelancer dataset, it turned on that the FairAlternatingGreedySetCoverAlgorithm and the RelaxedFairSetCoverRoundingAlgorithm outperform both the FairPaddingGreedySetCoverAlgorithm and the FairPairsGreedySetCoverAlgorithm in almost every case we considered, both in terms of solutions cost and in terms of solutions size. Overall, among these four algorithms, it seems more reasonable to opt for the RelaxedFairSetCoverRoundingAlgorithm. We can conclude that, even if the problem is not approximable in its weighted version, the algorithms we designed could be effectively used in practical contexts, and are able to find good solutions to many instances of the problem, at least in the limited case presented in the experiments chapter.  \\
Throughout this paper, we assumed that all workers in the workforce can be hired to complete each task; in other words, when creating a team for any task, algorithms can pick team members among all workers who make up the workforce: this is a pretty strong and unrealistic assumption since usually workers have a limited available time; therefore, in the future, it could be interesting to extend this research further by considering a stream of tasks, or by limiting the number of times each worker can be hired. Now, coherently with the scientific literature on Team Formation, another possibility worth of some consideration is the introduction of a social network among workers; this would lead to the emergence of new interesting research questions, such as finding a fair team that minimises the distance among workers.  Finally, to make this research more exhaustive, it could be convenient to study how the behaviour of Fair Team Formation algorithms changes with different datasets. Our experiments were, in fact, limited to tasks of no more than 6 skills, and to a workforce of only 1211 workers. We think that moving forward with this research could lead to some really useful and interesting results that, in turn, could help marketplaces designers, as well as policy makers, to better engineering, managing, and regulating these platforms.


\bibliographystyle{nourlabbrvnat}
\balance 
\bibliography{rent-or-buy}

\end{document}